\documentclass{scrartcl}

\usepackage{amsmath,amssymb,amsthm}
\usepackage{hyperref}

\newtheorem{theorem}{Theorem}
\newtheorem{lemma}[theorem]{Lemma}
\newtheorem{claim}[theorem]{Claim}

\newcommand{\cost}{\mathsf{cost}}

\begin{document}

\title{Improved and Simplified Inapproximability for $k$-means}
\author{Euiwoong Lee\thanks{Computer Science Department, Carnegie Mellon University, Pittsburgh, PA 15213}~\thanks{Supported by the Samsung Scholarship and NSF CCF-1115525.} \and Melanie Schmidt\footnotemark[1]~\thanks{Supported by the German Academic Exchange Service (DAAD).} \and John Wright\footnotemark[1]~\thanks{Supported by a Simons Fellowship in Theoretical Computer Science.}}

\maketitle
\begin{abstract}
The $k$-means problem consists of finding $k$ centers in $\mathbb{R}^d$ that minimize the sum of the squared distances of all points in an input set $P$ from $\mathbb{R}^d$ to their closest respective center. 
Awasthi et.\ al.\ recently showed that there exists a constant $\varepsilon' > 1$ such that it is NP-hard to approximate the $k$-means objective within a factor of $c$. We establish that the constant $\varepsilon'$ is at least $1.0013$. 
\end{abstract}

For a given set of points $P \subset \mathbb{R}^d$, the \emph{k-means problem} consists of finding a partition of $P$ into $k$ clusters $(C_1, \dots, C_k)$ with corresponding centers $(c_1, \dots, c_k)$ that minimize the sum of the squared distances of all points in $P$ to their corresponding center, i.e.\ the quantity
\[
\arg\min_{(C_1, \dots, C_k), (c_1, \dots, c_k)} \sum_{i=1}^k \sum_{x \in C_i} ||x-c_i||^2
\]
where $||\cdot||$ denotes the Euclidean distance.
The $k$-means problem has been well-known since the fifties, when Lloyd~\cite{L57} developed the famous local search heuristic also known as the $k$-means algorithm. Various exact, approximate, and heuristic algorithms have been developed since then. For a constant number of clusters $k$ and a constant dimension $d$, the problem can be solved by enumerating weighted Voronoi diagrams~\cite{IKI94}. If the dimension is arbitrary but the number of centers is constant, many polynomial-time approximation schemes are known. For example,~\cite{FL11} gives an algorithm with running time $\mathcal{O}(n d + 2^{\text{poly}(1/\varepsilon,k)})$. In the general case, only constant-factor approximation algorithms are known~\cite{JV01,KMNPSW04}, but no algorithm with an approximation ratio smaller than 9 has yet been found.

Surprisingly, no hardness results for the $k$-means problem were known even as recently as ten years ago. Today, it is known that the $k$-means problem is NP-hard, even for constant $k$ and arbitrary dimension $d$~\cite{ADHP09,D08} and also for arbitrary $k$ and constant $d$~\cite{MNV09}. Early this year, Awasthi et.\ al.~\cite{ACKS15} showed that there exists a constant $\varepsilon' > 0$ such that it is NP-hard to approximate the $k$-means objective within a factor of $1+\varepsilon'$. They use a reduction from the Vertex Cover problem on triangle-free graphs.  Here, one is given a graph $G=(V,E)$ that does not contain a triangle, and the goal is to compute a minimal set of vertices $S$ which \emph{covers} all the edges, meaning that for any $(v_i, v_j) \in E$, it holds that $v_i \in S$ or $v_j \in S$.
To decide if $k$ vertices suffice to cover a given $G$, they construct a $k$-means instance in the following way. Let $b_i=(0,\ldots,1,\ldots,0)$ be the $i$th vector in the standard basis of $\mathbb{R}^{|V|}$. For an edge $e=(v_i,v_j)\in E$, set $x_e=b_i+b_j$. The instance consists of the parameter $k$ and the point set $\{x_e \mid e \in E\}$. Note that the number of points is $|E|$ and their dimension is $|V|$. 

A relatively simple analysis shows that this reduction is approximation-preserving. 
A vertex cover $S \subseteq V$ of size $k$ corresponds to a solution for $k$-means where we have centers at $\{ b_i : v_i \in S \}$ and each point $x_{(v_i, v_j)}$ is assigned to a center in $S \cap \{ b_i, b_j \}$ (which is nonempty because $S$ is a vertex cover).  In addition, it can also be shown that a good solution for $k$-means reveals a small vertex cover of $G$ when $G$ is triangle-free. 

Unfortunately, this reduction transforms $(1 + \varepsilon)$-hardness for Vertex Cover on triangle-free graphs to $(1 + \varepsilon')$-hardness for $k$-means where $\varepsilon' = O(\frac{\varepsilon}{\Delta})$ and $\Delta$ is the maximum degree of $G$.
Awasthi et.\ al.~\cite{ACKS15} proved hardness of Vertex Cover on triangle-free graphs via a reduction from general Vertex Cover, where the best hardness result of Dinur and Safra~\cite{DS05} has an unspecified large constant $\Delta$. 
Furthermore, the reduction uses a sophisticated spectral analysis to bound the size of the minimum vertex cover of a suitably chosen graph product.

Our result is based on the observation that hardness results for Vertex Cover on small-degree graphs lead to hardness of Vertex Cover on triangle-free graphs with the same degree in an extremely simple way. 
Combined with the result of Chleb\'{i}k and Chleb\'{i}kov\'{a}~\cite{CC06} that proves hardness of approximating Vertex Cover on $4$-regular graphs within $\approx 1.02$, 
this observation gives hardness of Vertex Cover on triangle-free, degree-$4$ graphs without relying on the spectral analysis. 
The same reduction from Vertex Cover on triangle-free graphs to $k$-means then proves APX-hardness of $k$-means, with an improved ratio due to the small degree of $G$. 

\section{Main Result}
Our main result is the following theorem.
\begin{theorem}
It is NP-hard to approximate $k$-means within a factor $1.0013$. 
\end{theorem}

We prove hardness of $k$-means by a reduction from Vertex Cover on $4$-regular graphs, 
for which we have the following hardness result of Chleb\'{i}k and Chleb\'{i}kov\'{a}~\cite{CC06}.
\begin{theorem}[\cite{CC06}, see also~\ref{appendix:thm2}] 
Given a $4$-regular graph $G = (V(G), E(G))$, it is NP-hard to distinguish to distinguish the following cases.
\begin{itemize}
\item $G$ has a vertex cover with at most $\alpha_{min} |V(G)|$ vertices. 
\item Every vertex cover of $G$ has at least $\alpha_{max} |V(G)|$ vertices. 
\end{itemize}
Here, $\alpha_{min} = (2\mu_{4,k}+8)/(4 \mu_{4,k}+12)$ and $\alpha_{max} = (2\mu_{4,k}+9)/(4 \mu_{4,k}+12)$ with $\mu_{4,k} \le 21.7$. 
In particular, it is NP-hard to approximate Vertex Cover on degree-$4$ graphs within a factor of $(\alpha_{max}/\alpha_{min}) \geq 1.0192$. 
\end{theorem}

Given a $4$-regular graph $G = (V(G), E(G))$ for Vertex Cover with $n := |V(G)|$ vertices and $2n$ edges, 
we first partition $E(G)$ into $E_1$ and $E_2$ such that $|E_1| = |E_2| = |E(G)|/2 = n$ and such that the subgraph $(V(G), E_2)$ is bipartite. 
Such a partition always exists: every graph has a cut containing at least half of the edges (well-known; see, e.~g.,~\cite{MU05}). Choose $n$ of these cut edges for $E_2$, let $E_1$ be the remaining edges.
%
We define $G' = (V(G'), E(G'))$ by {\em splitting} each edge in $E_1$ into three edges.
Formally, $G'$ is given by  
\begin{align*}
& V(G') = V(G) \cup \left(\bigcup_{e = (u, v) \in E_1} \{ v'_{e, u}, v'_{e, v} \}\right), \\
& E(G') = \left(\bigcup_{e = (u, v) \in E_1} \left\{ (v, v'_{e, v}), (v'_{e, v}, v'_{e, u}), (v'_{e, u}, u) \right\}\right) \cup E_2 \enskip.
\end{align*}
Notice that $V$ has $n + 2n = 3n$ vertices and $3n + n = 4n$ edges. 
It is also easy to see that the maximum degree of $V$ is 4, and that $V$ does not have any triangle, since any triangle of $G$ contains at least one edge of $E_1$ (because $(V(G), E_2)$ is bipartite) and each edge of $E_1$ is split into three. 

Given $G'$ as an instance of Vertex Cover on triangle-free graphs, the reduction to the $k$-means problem is the same as before. 
Let $b_i=(0,\ldots,1,\ldots,0)$ be the $i$th vector in the standard basis of $\mathbb{R}^{3n}$. For an edge $e=(v_i,v_j)\in E(G')$, set $x_e=b_i+b_j$. The instance consists of the parameter $k = (\alpha_{min} + 1)n$ and the point set $\{x_e \mid e \in E\}$. Notice that the number of points is now $4n$ and their dimension is $3n$. 

We now analyze the reduction. 
Note that for $k$-means, once a cluster is fixed as a set of points, the optimal center and the cost of the cluster are determined\footnote{For $k=1$, the optimal solution to the $k$-means problem is the \emph{centroid} of the point set. This is due to a well-known fact, see, e.\ g., Lemma 2.1 in \cite{KMNPSW04}.}. 
Let $\cost(C)$ be the cost of a cluster $C$. 
We abuse notation and use $C$ for the set of edges $\{ e : x_e \in C \} \subseteq E(G')$ as well. 
For an integer $l$, define an \emph{$l$-star} to be a set of $l$ distinct edges incident to a common vertex. 
The following lemma is proven by Awasthi et.\ al.\ and shows that if $C$ is cost-efficient, then two vertices are sufficient to cover many edges in $C$. Furthermore, an {\em optimal} $C$ is either a star or a triangle. 

\begin{lemma}[\cite{ACKS15}, Proposition 9 and Lemma 11]\label{lem:cost}
Let $C = \{ x_{e_1}, \dots, x_{e_l} \}$ be a cluster. 
Then $l - 1 \leq \cost(C) \leq 2l - 1$, and there exist two vertices that cover at least $\lceil 2l - 1 - \cost(C) \rceil$ edges in $C$.
Furthermore, $\cost(C) = l - 1$ if and only if $C$ is either an $l$-star or a triangle, and otherwise, $\cost(C) \ge l- 1/2$.
\end{lemma}

\subsection{Completeness} 
\begin{lemma}
If $G$ has a vertex cover of size at most $\alpha_{min} n$, the instance of $k$-means produced by the reduction admits a solution of cost at most $(3 - \alpha_{min})n$.
\end{lemma}
\begin{proof}
Suppose $G$ has a vertex cover $S$ with at most $\alpha_{min} n$ vertices.
For each edge $e = (u, v) \in E_1$, let $v'(e) = v'_{e, u}$ if $v \in S$, and $v'(e) = v'_{e, v}$ otherwise. 
Let $S' := S \cup (\cup_{e \in E_1} \{ v'(e) \}$. Since $S$ is a vertex cover of $G$, for every edge $e \in E_1$, $S$ and $v'(e)$ cover all three edges of $E(G')$ corresponding to $e$. 
Therefore, $S'$ is a vertex cover of $G'$, and since $|E_1|=n$, it has at most $(\alpha_{min} + 1)n$ vertices. 

For the $k$-means solution, let each cluster correspond to a vertex in $S'$, and assign each edge $e \in E(G')$ to the cluster corresponding to a vertex incident to $e$ (choose an arbitrary one if there are two). 
Each edge is assigned to a cluster since $S'$ is a vertex cover, and each cluster is a star by construction.
Since there are $4n$ points and $k=\alpha_{min}n +n$, the total cost of the solution is, by Lemma~\ref{lem:cost},
\[
\sum_{i=1}^k \cost(C_i) = \sum_i^k (|C_i| - 1) = \bigg(\sum_i^k |C_i| \bigg) - k = (3 - \alpha_{min})n.\qedhere
\vspace*{-\baselineskip}
\]
\end{proof}

\subsection{Soundness} 
\begin{lemma}
If every vertex cover of $G$ has size of at least $\alpha_{max} n$, then any solution of the $k$-means instance produced by the reduction costs at least $(3 - \alpha_{min} + \frac{1}{3}(\alpha_{max} - \alpha_{min}))n$.
\end{lemma}
\begin{proof}
Suppose every vertex cover of $G$ has at least $\alpha_{max} n$ vertices. We claim that every vertex cover of $G'$ also has to be large. 
\begin{claim}
Every vertex cover of $G'$ has at least $(\alpha_{max} + 1)n$ vertices. 
\end{claim}
\begin{proof}
Let $S'$ be a vertex cover of $G'$. If $S'$ contains both $v'_{e, u}$ and $v'_{e, v}$ for any $e = (u, v) \in E_1$, 
then $S' \cup \{ u \} \setminus \{ v'_{e, u} \}$ is a vertex cover with the same or smaller size.
Therefore, we can without loss of generality assume that for each $e = (u, v) \in E_1$, $S'$ contains exactly one vertex in $\{ v'_{e, u}, v'_{e, v} \}$. 
Set $S := S' \cap V(G)$, thus $S$ has cardinality $|S'| - n$. 
Each $e \in E_2$ is covered by $S$ by definition. If an $e \in E_1$ is not covered by $S$, at least one of the three edges of $G'$ corresponding to $e$ is not covered by $S'$. 
Thus, every edge $e \in E(G)$ is covered by $S$, so $S$ is a vertex cover of $G$. Since $|S| \geq \alpha_{max}n$, $|S'| \geq (\alpha_{max} + 1)n$. 
\end{proof}


Fix $k$ clusters $C_1, \dots, C_k$. 
Without loss of generality, let $C_1, \dots, C_s$ be clusters that correspond to a star, and $C_{s + 1}, \dots, C_k$ be clusters that do not correspond to a star for any $l$. 
For $i = 1, \dots, s$, let $v(i)$ be the vertex covering all edges in $C_i$, and for $i = s+1, \dots, k$, let $v(i), v'(i)$ be two vertices covering at least $\lceil 2|C_i| - 1 - \cost(C_i) \rceil$ edges in $C_i$ by Lemma~\ref{lem:cost}. 
Let $E^{\dagger} \subseteq E(G')$ be the set of edges not covered by any $v(i)$ or $v'(i)$. 
The cardinality of $|E^{\dagger}|$ is at most 
\[
\sum_{i = s +1}^k (|C_i| - (2|C_i| - 1 - \cost(C_i))) = 
\sum_{i = s +1}^k (\cost(C_i) - (|C_i| - 1)). 
\]

Adding one vertex for each edge of $E^{\dagger}$ to the set $\{v(i)\}_{1 \leq i \leq s} \cup \{ v(i), v'(i) \}_{s + 1 \leq i \leq k }$ yields a vertex cover of $G'$ of size at most 
\[
s + 2(k - s) + \sum_{i = s +1}^k (\cost(C_i) - (|C_i| - 1)). 
\]
Every vertex cover of $G'$ has size of at least $(\alpha_{max} + 1)n = k + (\alpha_{max} - \alpha_{min})n$, so we have 
\[
(k - s) + \sum_{i = s +1}^k (\cost(C_i) - (|C_i| - 1)) \geq (\alpha_{max} - \alpha_{min})n. 
\]
Now, either $k - s \geq \frac{2}{3}(\alpha_{max} - \alpha_{min})n$ or $\sum_{i = s +1}^k (\cost(C_i) - (|C_i| - 1)) \geq \frac{1}{3}(\alpha_{max} - \alpha_{min})n$.
In the former case, since $\cost(C_i) \geq |C_i| - \frac{1}{2}$ for $i > s$ by Lemma~\ref{lem:cost}, the total cost is 
\[
\sum_{i=1}^k \cost(C_i) \geq \sum_{i=1}^s (|C_i| - 1) + \sum_{i=s+1}^k (|C_i| - \tfrac{1}{2}) \geq \bigg(\sum_i^k |C_i| \bigg) - k + \frac{(\alpha_{max} - \alpha_{min})n}{3}.
\]
In the latter case, the total cost can be split to obtain that  
$\sum\limits_{i=1}^k \cost(C_i) \geq \sum\limits_{i=1}^k (|C_i| - 1) + \sum\limits_{i=s+1}^k (\cost(C_i) - (|C_i| - 1))
  \geq \big(\sum\limits_i^k |C_i| \big) - k + \frac{1}{3}(\alpha_{max} - \alpha_{min})n.$
Therefore, in any case, the total cost is at least 
\[
\bigg(\sum_i^k |C_i| \bigg) - k + \frac{1}{3}(\alpha_{max} - \alpha_{min})n = \left(3 - \alpha_{min} + \frac{1}{3}(\alpha_{max} - \alpha_{min})\right)n.\qedhere
\vspace*{-\baselineskip}
\]
\end{proof}

The above completeness and soundness analyses show that it is NP-hard to distinguish the following cases.
\begin{itemize}
\item There exists a solution of cost at most $(3 - \alpha_{min})n$. 
\item Every solution has cost at least $(3 - \alpha_{min} + \frac{\alpha_{max} - \alpha_{min}}{3})n$.
\end{itemize}
Therefore, it is NP-hard to approximate $k$-means within a factor of
\[
\frac{(3 - \alpha_{min} + \frac{\alpha_{max} - \alpha_{min}}{3})n}{(3 - \alpha_{min})n} 
= 1 + \frac{\alpha_{max} - \alpha_{\min}}{3(3-\alpha_{min})}
= 1 + \frac{1}{3(10\mu_{4,k}+28)} 
\geq 1.0013.
\]

\bibliography{references-kmeanskt}
\bibliographystyle{amsalpha}

\appendix

\section{Remark on Theorem 2}\label{appendix:thm2}
To obtain Theorem 2, note that the proof of Theorem 17 in~\cite{CC06} states that it is NP-hard to distinguish whether the vertex cover has at most
\footnotesize
\[
|V(G)|\frac{2(|V(H)|-M(H))/k + 8 + 2\varepsilon}{2|V(H)|/k + 12} \text{ \normalsize or at least }
|V(G)|\frac{2(|V(H)|-M(H))/k + 9 + 2\varepsilon}{2|V(H)|/k + 12}
\]
\normalsize
vertices. By the assumption in the first sentence of the proof and because $|V(H)|=2M(H)$, $(|V(H)|-M(H))/k$ and $|V(H)|/k$ can be replaced by $\mu_{4,k}$ as defined in Definition 6 in~\cite{CC06}. By Theorem 16 in~\cite{CC06}, $\mu_{4,k} \le 21.7$.

\end{document}